\RequirePackage[l2tabu, orthodox]{nag}
\RequirePackage{snapshot}

\documentclass[10pt,onecolumn]{article}

\makeatletter
\if@twocolumn
  \usepackage[dvips,letterpaper,top=0.5in, bottom=0.5in, left=0.75in, right=0.5in,includefoot,heightrounded]{geometry}
\else
  \usepackage[dvips,letterpaper,margin=1in,includefoot,heightrounded]{geometry}
\fi

\usepackage{srcltx}

\usepackage{amsmath}
\usepackage{amssymb,amsfonts}

\usepackage{abstract}
\usepackage{algorithmic}

\usepackage{graphicx}
\usepackage{epsfig}
\usepackage[usenames,dvipsnames]{color}
\usepackage[usenames,dvipsnames]{xcolor}
\usepackage{subfigure}

\usepackage{booktabs}

\usepackage{setspace}
\usepackage{flushend}
\usepackage{multicol}

\usepackage{cite}
\usepackage{url}
\usepackage[normalem]{ulem}

\usepackage{enumerate}

\usepackage{hyperref}

\usepackage[sc,tiny,center]{titlesec}

\usepackage{microtype}

\newtheorem{theorem}{Theorem}



\title{%
A Generalized Prefix Construction
for OFDM Systems over Quasi-Static Channels
}

\author{%
Todor~Cooklev%
\thanks{%
Todor~Cooklev is with the Wireless Technology Center,
Indiana University-Purdue University Fort Wayne (IPFW), USA
(e-mail: cooklevt@ipfw.edu).
}
\and
Hakan~Do\u gan%
\thanks{%
Hakan~Do\u gan is with the Department of Electrical and Electronics
Engineering, Istanbul University,34320, Avcilar, Istanbul, Turkey
(e-mail: \mbox{hdogan@istanbul.edu.tr}).
}
\and
Renato~J.~Cintra%
\thanks{%
Renato~J.~Cintra is with
the Departamento de Estat\'istica, Universidade Federal de Pernambuco, Recife, Brazil
(e-mail: rjdsc@stat.ufpe.org).
}
\and
Hakan Y{\i}ld{\i}z%
\thanks{%
Hakan Y{\i}ld{\i}z is with the Ericsson, 34398, Maslak, Istanbul, Turkey
(e-mail: hakan.yildiz@ericsson.com).
}
}

\date{}

\newcommand{\myabstract}{%
All practical OFDM systems require a prefix to eliminate
inter-symbol interference at the receiver. Cyclic prefix~(CP) and
zero-padding are well-known prefix construction methods, the former
being
the most employed technique in practice due to its lower complexity.
In this paper we construct an OFDM system with a generalized CP. It
is shown that the proposed generalized prefix effectively
makes the channel experienced by the packet different from
the actual channel. Using an optimization procedure, lower bit
error rates can be achieved, outperforming other prefix construction
techniques. At the same time the complexity of the technique is
comparable to the CP method. The presented simulation
results show that the proposed technique not only outperforms the
CP method, but is also more robust in the presence of channel
estimation errors and mobility. The proposed method is
appropriate for practical OFDM systems.
}

\newcommand{\mykeywords}{%
Generalized prefix,
generalized skew-circular convolution,
\\
orthogonal frequency-division multiplexing (OFDM), deep fading channels.
}

\begin{document}

\makeatletter
\if@twocolumn

\twocolumn[%
  \maketitle
  \begin{onecolabstract}
    \myabstract
  \end{onecolabstract}
  \begin{center}
    \small
    \textbf{Keywords}
    \\\medskip
    \mykeywords
  \end{center}
  \bigskip
]
\saythanks

\else

  \maketitle
  \begin{abstract}
    \myabstract
  \end{abstract}
  \begin{center}
    \small
    \textbf{Keywords}
    \\\medskip
    \mykeywords
  \end{center}
  \bigskip
  \onehalfspacing
\fi

\section{Introduction}
Bandlimited wireless communication channels exhibit significant
variations in gain and phase along different frequencies. This
phenomenon reduces the capacity of the channel. In the case of
orthogonal frequency-division multiplexing (OFDM) systems,
associated subcarriers can experience deep fades. In the extreme
case these deep fades may become spectral nulls. When a cyclic
prefix~(CP) is used, these spectral nulls are known to significantly
limit the bit error rate performance.

In literature, there are several approaches to deal with the deep
fade problem. Practical OFDM systems~\cite{IEEE802.11WG} simply
accept this difficulty, and try to minimize it by the usage of
forward error correction~(FEC) techniques. However, applying FEC
increases complexity and overhead, and it may not be the right tool
for this task~\cite{wang2001linearly}. In fact, several signal
processing methods have been introduced that, without using FEC,
improve the performance in channels with spectral nulls. One such
method is the precoded OFDM, which inserts one zero between every
two information symbols~\cite{xia2001precoded}. Its disadvantage is
the reduction of data rate by one-third. In addition, precoded OFDM
requires the wireless channel to be static for the duration of
multiple OFDM symbols. This requirement may force the duration of
the OFDM block to be shorter than the channel coherence time.

Another approach is to replace the commonly used CP with
zero-padding (ZP)~\cite{scaglione1999redundant,muquet2002cyclic}.
Such a technique guarantees that the vector of information symbols
can be recovered regardless of the presence of spectral
nulls~\cite{scaglione1999redundant,muquet2002cyclic}. The reason is
that the equivalent channel matrix is invertible even if there are
spectral nulls. However, ZP has disadvantages. In particular, its
associated channel matrix is not circulant and cannot be
diagonalized by the discrete Fourier transform~(DFT), which is
generally computed via a fast Fourier transform~(FFT)
algorithm~\cite{blahut2010fast}. As a result, ZP is substantially
more complex than CP since it requires a matrix inversion for every
packet. Low-complexity ZP schemes have been developed, which however
are suboptimal~\cite{muquet2002cyclic}. Zero-padding also proves to
be ineffective in terms of timing and frequency
synchronization~\cite{andrews2007fundamentals}. As a result, neither
precoding, nor ZP are techniques employed in practical systems, such
as the IEEE~802 standards~\cite{IEEE802.11WG}. Indeed, all practical
systems use a CP because of its computational simplicity and
convenience for detection and synchronization.

In~\cite{muck2006pseudo} Muck~\emph{et al.} proposed
the insertion of a pseudo-random postfix~(PRP)
between OFDM symbols. This postfix is
constructed by means of a known vector weighted by
a pseudo-random scalar sequence to enable
semi-blind channel estimation.
As described in~\cite{muck2006pseudo},
at the receiver,
PRP-OFDM symbols are transformed into ZP-OFDM symbols by
subtracting the contribution of the postfix.
Subsequent additional processing is performed to
obtain a symbol corresponding to CP-OFDM.

{Other approaches that have been suggested are based on the
idea that the CP can be replaced by a pre-defined sequence of known
symbols} \cite{deneire2001},\cite{huemerunique}. {These
approaches are referred to as known symbol padding (KSP) OFDM and
unique word (UW) OFDM} \cite{Hofbauer}. { In these systems,
the known samples in the guard interval could be used for channel,
timing and carrier offset estimation.}

{UW-OFDM is a fundamentally different technique, which
introduces correlation among the subcarriers and the UW is part of
the DFT interval, whereas the prefix in CP-OFDM is added after the
DFT. This makes comparisons under identical assumptions difficult.
For example, the design of UW is considered subject to a constraint
on the peak-to-average power ratio (PAPR)}~\cite{coon2010designing}.
{In KSP-OFDM, the KSP part is not part of DFT, similar to
CP-OFDM, but performance simulations reported in}
\cite{steendam2007} {show that KSP-OFDM suffers from noise
enhancement and has a slightly worse performance than CP-OFDM.}

Our work is different:
we use a prefix like the CP, but weighted by
a complex number.
This new prefix can be referred to as a generalized prefix.
The proposed technique effectively changes the phases of the multipath
components.
We show that if the prefix is appropriately constructed
the technique can effectively transform the wireless channel
experienced by the OFDM symbols into a higher-capacity channel,
leading to a lower bit-error rate.
The described prefix can convert
a channel with spectral nulls or deep fades into
a channel with no spectral nulls or fades that are less deep.
{Furthermore,
if all other parameters are identical,
it outperforms several prefix
construction techniques
at the cost of a minor increase in the computational complexity.}
We present extensive simulation results.
The proposed technique outperforms
the CP method, particularly in channels with long delay spreads that are characterized by deep fades.
Furthermore, the proposed system is more robust to
channel estimation errors, i.e. compared with CP-OFDM the margin of improvement is even higher
in the presence of channel estimation errors.

The organization of the paper is as follows.
In Section~\ref{section.signal} we formulate the problem mathematically.
A generalized prefix is introduced and
discussed in Section~\ref{section.generalized}.
Next we detail how the proposed generalized prefix
can be optimized and the associated computational complexity
is analyzed in Section~\ref{section.optimal}.
Section~\ref{section.simulation} contains
the simulation results and concluding remarks are drawn in Section~\ref{section.conclusions}.

\section{Signal Model}
\label{section.signal}

We consider a conventional OFDM system with
single transmitting and receiving antennas.
The complex information symbols, referred to as subcarriers,
are denoted by $X[k]$, $k = 0,1, \ldots, N-1$,
where $N$ is the number of subcarriers.
These complex symbols are obtained from
the information bits via a bit-to-symbol mapping such as
quadrature phase-shift keying,
or more generally $M$-ary quadrature amplitude modulation~(QAM).
Without any limitation to the proposed methodology,
hereafter we assume that the QAM scheme is employed.

An application of the inverse DFT generates a new signal~$\mathbf{x}$
that combines all subcarriers:
\begin{equation}
\label{1}
\mathbf{x}
=
\mathbf{F}_N^*
\cdot
\mathbf{X},
\end{equation}
where $\mathbf{F}_N$ is the $N$-point DFT matrix which is defined as
$\mathbf{F}_N \triangleq \left[ \exp(-2\pi ik /N)
\right]_{i,k=0,\ldots,N-1}$, superscript~${}^\ast$ denotes Hermitian
conjugation, and $\mathbf{X} = \big[ X[0], X[1], \ldots, X[N-1]
\big]^t$.
Next, cyclic prefix (CP) of length $K$ is added to the OFDM
symbols~$\mathbf{x}$ to obtain a new vector
$\widetilde{\mathbf{x}}$:
\begin{equation}
\begin{split}
\widetilde{\mathbf{x}} \triangleq \big[ x[N-K],  x[N-K-1], \ldots,
x[N - 1], x[0], x[1], \ldots, x[N - 1] \big]^t ,
\end{split}
\end{equation}
which has $N+K$ components.

A wireless channel with $L$ multipath components can be modeled with respect to
the baseband by an $L$-point FIR filter,
whose transfer function is
\begin{equation}
H(z) = h[0] + h[1] z^{-1} + \cdots + h[L-1] z^{-L+1}
.
\end{equation}
The CP prevents ISI,
as long as its length exceeds the duration of the impulse response of the channel,
or $K \ge L - 1$.
Then
the received signal is a linear convolution between
the transmitted signal and the channel impulse response for
$m = 0, 1, \ldots, N+K+L-2$:
\begin{equation}
y[m]
=
\sum_{l= 0}^{L-1}
h[l]
\widetilde{x}[m-l] + w[m],
\quad
\end{equation}
where  $w[m]$ is Gaussian noise vector with zero mean and
variance $\sigma^2 = N_0/2$
and $N_0$ is the single-sided power spectral density.
Over the first $K$ received symbols there is inter-symbol interference,
and
these symbols are ignored.
The receiver processes $y[m]$ only for $m = K, K + 1, \ldots, N + K - 1$,
or alternatively
\begin{align}
y[m]
=
\sum_{l = 0}^{L - 1}
h[l]
x[\left\langle m - K - l \right\rangle_N]
+ w[m],
\end{align}
where $\langle \cdot \rangle_N$ represents the modulo $N$ operation.
The CP removal not only eliminates the ISI, but also converts the
linear convolution with the channel impulse response into a cyclic
convolution. Consider the $N$-point signal $\mathbf{y} = \left[
y[K], y[K + 1], \ldots, y[N+K-1] \right]^t$ and the vector of
channel impulse response components $\mathbf{h} = \left[ h[0], h[1],
\ldots, h[L-1], 0, \ldots, 0 \right]^t$ padded with an appropriate
number of zeros, so that its length becomes $N$. By taking~(\ref{1})
into account, the referred cyclic convolution can be expressed in
matrix terms according to
\begin{equation}\label{6}
\mathbf{y} = \mathbf{H} \cdot \mathbf{F}_N^* \cdot \mathbf{X} + \mathbf{w},
\end{equation}
where $\mathbf{w} = \left[ w[K], w[K+1], \ldots, w[N+K-1]\right]^t$
and $\mathbf{H}$ is the circulant matrix whose first column
is the vector $\mathbf{h}$:
\begin{align}
\mathbf{H}
=
\left[
\begin{smallmatrix}
h[0]   &  0     & \cdots &  0      & h[L-1]  & h[L-2] & \cdots & h[1]   \\
h[1]   &  h[0]  & 0      &  \ddots & 0       & h[L-1] & \ddots & h[2]   \\
\vdots & \ddots & \ddots &  \ddots & \ddots  & \ddots & \ddots & \vdots \\
0      & \cdots & 0      & h[L-1]  & h[L-2]  & \cdots & h[1]   & h[0]
\end{smallmatrix}
\right]
.
\end{align}
It is assumed that the channel is characterized by slow fading,
i.e. the channel impulse response does not change within one OFDM symbol duration.
At the receiver shown in Fig.~\ref{fig:rec2},
the forward DFT is applied to obtain
\begin{equation}
\begin{split}
\mathbf{Y} &= \mathbf{F}_N \cdot \mathbf{y}
= \mathbf{F}_N \cdot \mathbf{H} \cdot \mathbf{F}_N^* \cdot \mathbf{X}
+
\mathbf{F}_N \cdot \mathbf{w}
.
\end{split}
\end{equation}
The information symbol vector $\mathbf{X}$ must be recovered from $\mathbf{Y}$.
This can be accomplished by the multiplication of $\mathbf{Y}$  with the inverse of
$\mathbf{F}_N \cdot \mathbf{H} \cdot \mathbf{F}_N^*$.
This matrix inversion is greatly simplified numerically,
because the DFT diagonalizes circulant matrices:
\begin{equation}
\label{9}
\begin{split}
\mathbf{F}_N \cdot \mathbf{H} \cdot \mathbf{F}_N^*
&=
\mathrm{diag}
\left(
H[0], H[1], \ldots, H[N-1]
\right)
\\
&= \mathrm{diag}\left( \mathbf{F}_N \cdot \mathbf{h} \right)
,
\end{split}
\end{equation}
where $\mathrm{diag}(\cdot)$ returns a diagonal matrix
and
the sequence $H[k]$, $k=0,1,\ldots,N-1$, is the $N$-point DFT
of the channel impulse response.

\begin{figure}%
\centering
\input{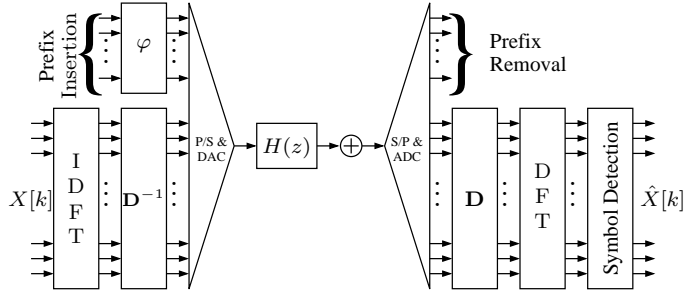}
\caption{Proposed OFDM system.}
\label{fig:rec2}
\end{figure}

The diagonalization operation effectively decomposes
the channel into
parallel and ISI-free sub-channels.
In other words,
the frequency-selective channel is transformed into
a channel with flat fading per subcarrier.
As a result,
a simple zero-forcing (ZF) detector
can be employed
to obtain an estimate $\widehat{\mathbf{X}}$
of the transmitted information sequence.
Requiring only one division operation per subcarrier,
the elements of $\widehat{\mathbf{X}}$ are given by
\begin{equation}\label{10}
\widehat X[k] = \frac{Y[k]}{H[k]},
\quad
k = 0,1,\ldots,N-1
.
\end{equation}

The main disadvantage of the described OFDM system is the
performance in channels with spectral nulls. Indeed, if one or more
of the quantities $H[k]$, $k=0,1,\ldots,N-1$, are equal to zero,
then~(\ref{10}) yields to an indetermination. Equivalently, the
diagonal matrix $\mathrm{diag}\left( \mathbf{F}_N \cdot \mathbf{h}
\right)$ is not invertible.

Let the probability of bit error for the chosen QAM bit-to-symbol
mapping scheme in an AWGN channel be $P_\text{QAM}\left( E_b/N_0
\right)$, where $E_b$ is the energy per bit. Then the probability of
bit error $P_e$ of the OFDM system with cyclic prefix
is~\cite{xia2001precoded}:
\begin{equation}
\label{11}
P_e
=
\frac{1}{N}
\sum_{k = 0}^{N-1}
P_\text{QAM}
\left(
\frac{N}{N + K} \frac{E_b}{N_0}
|H[k]|^2
\right)
.
\end{equation}
Notice that the quantity $\frac{N}{N + K} \frac{E_b}{N_0}$ is a constant.

If the channel has one spectral null, or $H[k]=0$ for a given
subcarrier~$k$, then the symbol transmitted on that subcarrier has a
probability of error equal to $1/2$, regardless of the
signal-to-noise ratio. Since the probability of error of the entire
system is the average of the probabilities of error for each of the
subcarriers, it is known that the subcarriers with the largest
probability of error dominate the overall system error probability.
Since there are $N$ subcarriers, the probability of error will
always be higher than  $P_e \ge 1/(2N)$. In other words, there is an
error floor for high signal-to-noise values, or equivalently, there
is a loss in frequency diversity~\cite{scaglione1999redundant}. The
effect of deep fades will be similar, only less severe. This is a
significant disadvantage.

As mentioned in the introduction, ZP is another possible prefix
construction technique. After the receiver discards all received
samples, except those from $K$ to $N+K-1$, ZP leads to a received
signal as described in~(\ref{6}) with the equivalent channel matrix
being equal to:
\begin{equation}
\mathbf{H}_\text{ZP}
=
\left[
\begin{smallmatrix}
h[0]    &  0     & \cdots &  0      & 0  & \cdots & 0   & 0    \\
h[1]    &  h[0]  & 0      &  \ddots & \ddots  & \ddots & \ddots & 0      \\
 \vdots & \ddots & \ddots &  \ddots & \ddots  & \ddots & \ddots & \vdots \\
0       & \cdots & 0      & h[L-1]  & h[L-2]  & \cdots & h[1]   & h[0] \\
0       & \cdots & 0      & 0       & h[L-1]  & \cdots & h[2]   & h[1] \\
\vdots  & \vdots & \vdots & \vdots  & \vdots  & \vdots & \vdots & \vdots
\end{smallmatrix}
\right]
.
\end{equation}
This matrix is similar to~$\mathbf{H}$;
however it is not necessarily square.
Indeed, it is a tall channel matrix.

The matrix~$\mathbf{H}_\text{ZP}$ cannot be diagonalized by applying
the forward DFT. As a result, the receiver has to perform an
{$N\times N$ matrix inversion for every OFDM symbol}
to obtain an estimate
$\widehat{\mathbf{X}} $ of the information symbol vector:
\begin{equation}
\widehat{\mathbf{X}}
=
\mathbf{F}_N \cdot \mathbf{H}^{+}_\text{ZP} \cdot \mathbf{y}
=
\mathbf{X} + \mathbf{F}_N \cdot \mathbf{H}^{+}_\text{ZP} \cdot \mathbf{w}
,
\end{equation}
where superscript ${}^+$ denotes the pseudoinverse operation.

{The matrix $\mathbf{H}_\text{ZP}$ is always full column rank
---
even in the presence of spectral nulls
---
as long as $h[0]\neq 0$, which is a reasonable assumption.
However, still the ZP method is not used by any practical system due
to its high computational complexity.}

\section{Proposed Generalized Prefix}
\label{section.generalized}

In the proposed system the prefix is constructed by multiplying the
last $K$ samples with a complex number $\varphi  \ne 0$
and appending them to $\mathbf{x}$:
\begin{equation}\label{13}
\begin{split}
\widetilde{\mathbf{x}}
=
\big[
\varphi x[N-K],  \varphi x[N-K-1], \ldots, \varphi x[N - 1],
x[0], x[1], \ldots, x[N - 1]
\big]
.
\end{split}
\end{equation}

Just like in the case for CP, the receiver processes $y[m]$ only for
$m=K, K+1, \ldots, N+K-1$. The received signal can be expressed as
\begin{equation}\label{14}
y[m]
=
\sum_{l=0}^{L-1}
h[l]
u[m-K-l]
x[\left\langle m-K-l \right\rangle_N]
+ w[m],
\end{equation}
where
the sequence $u[\cdot]$ is
\begin{equation}
u[n] =
\begin{cases}
   1, & \text{if $n \geq 0$,}  \\
   \varphi,  & \text{if $n < 0$.}  \\
\end{cases}
\end{equation}

It is worth pointing out that~(\ref{14}) has not been defined as
a convolution operation in the signal processing
literature~\cite{oppenheim2009discrete,martucci1994symmetric,ersoy1985semisystolic}.
The special cases $\varphi  = 1$  and $\varphi  = -1$
correspond to circular and skew-circular
convolution~\cite{martucci1994symmetric,ersoy1985semisystolic}.
We suggest to call the particular type of convolution
described in~(\ref{14}) as generalized skew-circular convolution.

Considering the proposed prefix,
the equivalent channel matrix is given by:
\begin{equation}
\begin{split}
\mathbf{H}_\varphi
=
\left[
\begin{smallmatrix}
h[0]   &  0     & \cdots &  0      & \varphi h[L-1]  & \varphi h[L-2] & \cdots & \varphi h[1]  \\
h[1]   &  h[0]  & 0      &  \ddots & 0               & \varphi h[L-1] & \ddots & \varphi h[2]  \\
\vdots & \ddots & \ddots &  \ddots & \ddots          & \ddots         & \ddots & \vdots        \\
0      & \cdots & 0      &  h[L-1] & h[L-2]          & \cdots         & h[1]   & h[0]
\end{smallmatrix}
\right]
.
\end{split}
\end{equation}
The matrix $\mathbf{H}_\varphi$ can be diagonalized. Let us denote
the elements of $\mathbf{H}_\varphi$ as $h_{ij}^\varphi$,
$i,j=1,2,\ldots,N$.
We note that $\mathbf{H}_\varphi$ is a Toeplitz matrix
with the property $h_{ij}^\varphi  = \varphi h[N+i-j]$
for $i-j<0$.
For $i-j\geq0$,
$\mathbf{H}_\varphi$ is constructed as $\mathbf{H}$.
We suggest that this matrix type is referred to as generalized
skew-circular, following the name of the convolution operation.
The following result allows the use of the proposed prefix in OFDM systems.

\begin{theorem}
If $\psi$ is a complex number for which
\begin{equation}\label{17}
\psi^N = \varphi,
\end{equation}
then
\begin{equation}
\begin{split}
\mathbf{F}_N
\cdot
\mathbf{D}
\cdot
&
\mathbf{H}_\varphi
\cdot
\mathbf{D}^{-1}
\cdot
\mathbf{F}_N^\ast
=
\mathrm{diag}
\left( H_\psi[0], H_\psi[1], \ldots, H_\psi[N-1] \right)
=
\mathrm{diag}\left(\mathbf{F}_N \cdot \mathbf{D} \cdot \mathbf{h} \right),
\end{split}
\end{equation}
where $H_\psi [k]$, $k=0,1, \ldots, N-1$ is the DFT of the vector
$\mathbf{D} \cdot \mathbf{h}$, and
\begin{equation}\label{19}
\mathbf{D}=\mathrm{diag} \left( 1, \psi, \psi^2, \ldots, \psi^{N-1} \right).
\end{equation}
\end{theorem}

\begin{proof}
Suppose that the $(i,j)$ entry of
$\widetilde{\mathbf{H}} = \mathbf{D} \cdot \mathbf{H}_\varphi \cdot \mathbf{D}^{-1}$
is denoted by $\widetilde{h}_{ij}$.
Considering the DFT property in~(\ref{9}),
the above result can be proven by establishing that
$\mathbf{D} \cdot \mathbf{H}_\varphi \cdot \mathbf{D}^{-1}$
is a circulant matrix.
Circulant matrices have the property $\widetilde{h}_{ij} = \widetilde{h}_{kl}$
when  $k - l = \left\langle i-j \right\rangle_N$.
Thus, the following holds:
\begin{equation} \label{20}
\begin{split}
\widetilde{h}_{ij}
=&
\begin{cases}
\psi^{i-1} h[i-j] \psi^{-j+1},  & \text{if $i-j\geq0$,} \\
\psi^{i-1} \psi^N h[N+i-j] \psi^{-j+1},  & \text{if $i-j<0$}
\end{cases}
\\
=&
\begin{cases}
\psi^{i-j}   h[i-j],   & \text{if $i-j\geq0$,} \\
\psi^{N+i-j} h[N+i-j], & \text{if $i-j<0$.}
\end{cases}
\end{split}
\end{equation}

Considering~(\ref{21}) and $h_{ij}^\varphi  = \varphi h[N+i-j]$
for $i-j<0$, if  $k-l = \left\langle i-j \right\rangle_N$,
then it follows that $\widetilde{h}_{ij} = \widetilde{h}_{kl}$ .
Therefore the matrix
$\widetilde{\mathbf{H}}$ is circulant.
Additionally,
from~(\ref{20}) we notice that the first column of
$\widetilde{\mathbf{H}}$
is the vector
$\mathbf{D}\cdot \mathbf{h}$:
\begin{equation}
\label{21}
\widetilde{\mathbf{H}}
=
\left[
\begin{smallmatrix}
h[0]      &  0     & \cdots & 0                & \psi^{L-1}h[L-1]  & \psi^{L-2} h[L-2] & \cdots    & \psi h[1]   \\
\psi h[1] &  h[0]  & 0      & \ddots           & 0                 & \psi^{L-1} h[L-1] & \ddots    & \psi^2 h[2] \\
\vdots    & \ddots & \ddots & \ddots           & \ddots            & \ddots            & \ddots    & \vdots     \\
0         & \cdots & 0      & \psi^{L-1}h[L-1] & \psi^{L-2}h[L-2]  & \cdots            & \psi h[1] & h[0]
\end{smallmatrix}
\right]
.
\end{equation}
\end{proof}
A comparable result was proven in~\cite{muck2006pseudo} for the postfix OFDM.

The block diagram of the proposed system is shown in Fig.~\ref{fig:rec2}. Multicarrier modulation is performed by
\begin{equation}
\mathbf{x} = \mathbf{D}^{-1} \cdot \mathbf{F}_N^* \cdot \mathbf{X}.
\end{equation}

At the receiver side the signal is multiplied by $\mathbf{D}$, and then the DFT is applied to obtain
\begin{equation}
\begin{split}
\mathbf{Y}
&=
\mathbf{F}_N \cdot \mathbf{D} \cdot \mathbf{H}_\varphi \cdot \mathbf{D}^{-1} \cdot \mathbf{F}_N^\ast \cdot \mathbf{X}
+
\mathbf{F}_N \cdot \mathbf{D} \cdot \mathbf{w}
\\
&=
\mathrm{diag}
\left(H_\psi[0], H_\psi[1], \ldots, H_\psi[N-1]
\right)
\cdot
\mathbf{X}
+
\mathbf{F}_N \cdot \mathbf{D} \cdot \mathbf{w}
.
\end{split}
\end{equation}

Since $\mathbf{D}$ is a diagonal matrix, vector multiplications
by~$\mathbf{D}$ or~$\mathbf{D}^{-1}$ require only $N$ complex
multiplications. Therefore the arithmetic complexity of
modulation/demodulation operation is comparable to the CP case. On
the other hand, the proposed prefix requires an optimization
procedure to find the optimal value of $\psi$; this also requires an
additional complexity.

In general, the quantities $\varphi$ and $\psi$ can be any non-null
complex numbers that satisfy~(\ref{17}). However, if $| \varphi |
\ne 1$, then the absolute values of the elements of $\mathbf{D}$ and
$\mathbf{D}^{-1}$ are either very large or very small. This fact
makes the system prone to numerical instability issues.
Additionally, the PAPR is increased. To maintain the same PAPR, it
is required that $|\varphi| = 1$. This choice of $\varphi$ also
implies $|\psi|=1$.

\section{{Generalized} Prefix}
\label{section.optimal}

The proposed prefix construction transforms the wireless channel
$\mathbf{h}$ experienced by the system into an equivalent channel
with impulse response equal to the product $\mathbf{D} \cdot
\mathbf{h}$, whose transfer function is given by:
\begin{equation}
\label{24}
H_\psi(z)
=
h[0] +
\psi h[1] z^{-1} +
\psi^2 h[2] z^{-2} +
\cdots +
\psi^{L-1} h[L-1] z^{-(L-1)}
.
\end{equation}

It is now clear that the generalized prefix has a physical meaning:
the $n$th multipath component is multiplied by $\psi^n$. By adopting
$|\psi|=1$, we can write $\psi = e^{j\alpha}$, for
$\alpha\in[0,2\pi]$; thus only the phase of the multipath components
is affected.

The role of $\psi$ is better understood in terms of the
frequency response $H_\psi(e^{j \omega})$ of the equivalent channel,
which can be related to the frequency response $H(e^{j\omega})$ of
the actual channel:
\begin{equation}
\begin{split}
H_\psi(e^{j \omega})
=
\sum_{n=0}^{N-1}
\psi^n
h[n]
e^{-j \omega n}
=
H(e^{j(\omega-\alpha)})
,
\end{split}
\end{equation}
since $\psi = e^{j\alpha}$. Thus, the proposed prefix construction
represents a shift in frequency domain of the channel frequency
response.

Since the subcarriers are centered at the discrete frequencies
$\omega_k = 2\pi k/N$, $k=0,1,\ldots,N-1$, it is important to
characterize the associated quantities $H[k] = H(e^{j\omega_k})$.
Suppose that the channel frequency response
has a null at $\omega_k$, for some $k$. In this case, the frequency
response of the channel, actually experienced by the system
can avoid this situation by shifting the spectral null
away from the location $\omega_k$.

It must be noted that $\varphi = \psi^N = 1$ or  $\varphi = \psi^N =
-1$, which correspond to circular or skew-circular convolution,
respectively, may not be adequate choices. In fact, this selection
of $\varphi$ entails $\psi = e^{j2\pi/N}$ or $\psi = e^{j\pi/N}$.
Thus, the implied shift in frequency is $\alpha = 2\pi/N$ or $\alpha
= \pi/N$, respectively. Since the subcarriers are exactly $2\pi/N$
rad/samples apart from each other, the spectral shift generated by
circular or skew-circular convolution may simply move an existing
zero of the sampled transfer function from one subcarrier to
another. This would not be effective. In order to appropriately
avoid spectral nulls, it is necessary to have a greater freedom of
choice for $\psi$.

The distinct advantage of the generalized prefix is that for
certain values of $\psi$ the wireless channel experienced
by the transmitted OFDM symbols can lead to a lower bit error rate.
In principle, any choice of $\alpha\neq 2\pi m /N$, where $m$ is an
integer, is likely to improve the performance of the above described
multicarrier modulation system up to certain extent. Even a random
selection scheme for the value of $\psi$ could represent an
enhancement in avoiding spectral nulls. However, there exists an
optimal choice of $\psi$ that leads to the lowest bit error rate,
therefore increasing system performance.
Next we discuss two methods of finding the optimal value of $\psi$.

\subsection{Minimal Probability of Error Approach}

In the first approach, we consider the probability of bit error the
figure of merit to be optimized.
Thus, the minimization problem is:
\begin{align}
\label{optim.problem.exact}
\alpha^\ast
=
\arg
\min_{\alpha \in \left[0, \frac{2\pi}{N} \right]}%
P_e(\psi)|_{\psi=e^{j\alpha}},
\end{align}
where
$\alpha^\ast$ is the optimum frequency shift
and $P_e(\psi)$ is the bit error probability of
the  proposed system
given by
\begin{equation}\label{26}
P_e(\psi) =
\frac{1}{N}\sum_{k=0}^{N-1}
P_\text{QAM}
\left(
\frac{N}{N + K} \frac{E_b}{N_0}
| H_\psi[k] |^2
\right)
.
\end{equation}
The above equation is analogous to (\ref{11}).
We assume that $P_e(\cdot)$ is continuous and unimodal over the
search space for a fixed $E_b/N_0$.
Additionally, $P_e(\cdot)$ is a single-variable bounded nonlinear
function. Because the frequency separation between adjacent
subcarriers is $2\pi/N$ the search space can be $[0, 2\pi/N]$, and
the set of frequency shifts $\{\alpha^\ast + 2\pi m /N \}$, where
$m$ is an integer, constitutes a class of equivalence.

Under the above conditions,
the suggested minimization problem can be solved
without resorting to derivatives
by means of numerical techniques such as
the golden section search method.
Let us assume that $[a,b]$ is a given interval where the sought
minimum is located. In this case, the points $p$ and $q$ can be
calculated as follows:
\begin{equation}
\begin{split}
p &= b - (b-a) \Phi, \\
q &= a + (b-a) \Phi,
\end{split}
\end{equation}
where $\Phi=\frac{\sqrt{5}-1}{2}$ is the golden ratio conjugate.
After evaluating the function $P_e(\cdot)$ at points $p$ and $q$, a
new search interval is established according to whether $P_e(e^{jp})
\leq P_e(e^{jq})$ or not \cite{geraldapplied}. The algorithm is
summarized in Fig.~\ref{golden}.

\begin{figure}
\hrulefill

\footnotesize
   \textbf{Input:} Initial interval $[a,b]$; tolerance $\epsilon$; objective function $P_e(\cdot)$. \\
   \textbf{Output:} Optimum frequency shift $\alpha^\ast$. \\
   \textbf{Method:} Golden section search. \\ [-0.45cm]

\hrulefill

\begin{algorithmic}[1]

\STATE $\Phi \leftarrow (\sqrt{5}-1)/2$
\STATE $p \leftarrow b
-(b-a) \Phi$
\STATE $q \leftarrow a + (b-a) \Phi$
\STATE $f_p
\leftarrow P_e(e^{j p})$
\STATE $f_q \leftarrow P_e(e^{j q})$

\WHILE{$b-a \ge \epsilon$}

\IF{$f_p \leq f_q$}
  \STATE $b \leftarrow q$
  \STATE $q \leftarrow p$
  \STATE $p \leftarrow b - (b-a) \Phi$
  \STATE $f_q \leftarrow f_p$
  \STATE $f_p \leftarrow P_e(e^{j p})$
\ELSE
  \STATE $a \leftarrow p$
  \STATE $p \leftarrow q$
  \STATE $q \leftarrow a + (b-a) \Phi$
  \STATE $f_p \leftarrow f_q$
  \STATE $f_q \leftarrow P_e(e^{j q})$
\ENDIF

\ENDWHILE

\RETURN $\alpha^\ast  \leftarrow (a + b)/2$
\end{algorithmic}
\hrulefill \caption{Minimization algorithm: golden section search.}
\label{golden}
\end{figure}

\subsection{Max-min Approach}

Searching for the optimal value of $\alpha$ by solving the
optimization problem as described in (\ref{optim.problem.exact})
involves several calls to the $Q$-function and other arithmetical
operations. Also the function $P_e(\cdot)$ depends on the chosen
modulation order.

Instead of~(\ref{optim.problem.exact}), as an alternative, we can
consider the following optimization problem:
\begin{align}
\label{optim.problem.alternative}
\alpha^\ast= \max_{\alpha \in \left[0, \frac{2\pi}{N} \right]}
\min_k |H_\psi[k]|_{\psi=e^{j\alpha}}.
\end{align}
Intuitively,
the above problem returns the value of $\alpha$ (therefore, $\psi$)
that makes the minimum values of $H_\psi[k]$
(deep fades) to be as large as possible.

Our simulations indicate that solving the above minimax type
optimization problem~\cite{drezner1982minimax} leads to the same
numerical results as the optimization technique in
from~(\ref{optim.problem.exact}). The discussed above golden search
method is also employed in this procedure.

\subsection{Practical Considerations}

An important feature of the proposed technique is that the optimal
value can be determined not only at the transmitter (if the
transmitter has channel state information~(CSI)), but also at the
receiver. Although making~CSI available at the transmitter is
difficult for both time division duplex and frequency division
duplex modes~\cite{rong2006adaptive} recent wireless systems such as
802.11 can provide CSI at the transmitter~\cite{IEEE802.11WG}.

Still, a considerable amount of information must be sent back to the
transmitter, especially for mobile systems. Therefore, we propose
that $\psi$ is calculated at the receiver and only the value of
$\psi$ is sent to the transmitter. In this case, feedback can be
provided with minimum overhead and very quickly to match the channel
variability.

\subsection{Computational Complexity}

In the proposed system the transmitter is required to compute $2N+K$
complex multiplications~(CMs): $K$ CMs to generate the prefix
in~(\ref{13}); $N$ CMs for $\mathbf{D}^{-1}$; and $N$ CMs to perform
multiplication by $\mathbf{D}^{-1}$ according to~(\ref{19}).
At the receiver, the diagonal matrix~$\mathbf{D}$ must also be
obtained from the value of $\psi$ and the received signal must be
multiplied by~$\mathbf{D}$; each of these operations requires $N$
CMs. Now let us evaluate the complexity of finding the optimal value
of $\psi$.

\subsubsection{Probability Based Optimization Approach}

First we analyze the method described in (\ref{optim.problem.exact}).
To find the needed optimal $\psi$ for~$\mathbf{D}$, the receiver
calculates the associated cost function by using~(\ref{24})
and~(\ref{26}) in the golden search algorithm.
We assume that this optimization algorithm
takes $Z$ iterations to converge.
Thus,
for each iteration,
we have the following computational complexities.

The evaluation of~(\ref{24}) requires only~$L-1$ CMs.
An $N$-point DFT call is required
to provide the magnitude of the frequency response
of the equivalent channel
as shown in~(\ref{26}).
The multiplicative complexity of the DFT is in~$\mathcal{O}(N \log_{2}N)$.
The evaluation of function~$P_\text{QAM}$ in~(\ref{26})
requires the $Q$-function~\cite{szczecinski2006arbitrary}.
Assuming that the~$Q$-function values are stored in a lookup table,
(\ref{26}) requires
$N$ real multiplications by a constant in the argument of~$P_\text{QAM}$
and
an $N$-point summation.

Therefore, the proposed prefix based system increases the
computational complexity by $2N+K + 2N + Z \times (N + L - 1 + N
\log_2 N)$ CMs.

\subsubsection{Max-min Optimization Approach}

On the other hand,
the described max-min method
as shown in (\ref{optim.problem.alternative})
is computationally less expensive.
Indeed,
in comparison to~(\ref{optim.problem.exact}),
it eliminates
(i)~the evaluation $P_\text{QAM}$;
(ii)~the $N$-point summation;
and
(iii)~the required multiplications by constants.
Since the function $P_\text{QAM}$ is not needed, it also does not
require any sort of lookup table.

It requires
(i) the evaluation of (\ref{24}) ($L-1$ CMs);
(ii) a DFT call;
and
(iii) the identification of the smallest element in a vector.
The DFT computation is in
$\mathcal{O}(N \log_2 N)$.
Finding the minimum in a vector is comparatively
a simple operation~\cite{cormen2001algorithms}
(cf.~(\ref{optim.problem.exact})).
Thus,
we have a total multiplicative complexity
of
$2N + K + 2N + Z(L-1 + N \log_2 N)$.
Due to its simplicity,
the max-min method may be regarded as
computationally faster than the first approach.

For both methods, the above detailed estimates of the number of
multiplications are worst case scenarios. This is due to the fact
that we are assuming that the DFT computation would require
exactly~$N \log_2 N$ multiplications, which is an upper bound. For
example, assuming~$N=512$ and using the Rader-Brenner radix-two fast
Fourier transform~\cite{blahut2010fast}, the DFT can be computed in
just 3076 real multiplications, which is well below $512 \times
\log_2 512 = 4608$~CMs (13824 real
multiplications)~\cite{blahut2010fast}.

For $N \gg L$ and $N \gg K$, as is the case in practical systems,
the computational complexity is $\mathcal{O}(N)$ at the transmitter.
At the receiver side, we have an asymptotical computational
multiplicative complexity in $\mathcal{O}(Z N \log N)$ for both
methods.

{Overall the computational complexity of the proposed method
is higher than the complexity of CP-OFDM for both optimization
techniques. However, $\psi$ can be calculated at the beginning and
can be used while the channel is assumed to be static. Therefore, we
can say that at the expense of a minor increase in the computational
complexity we can achieve a performance gain which will be shown in
the simulation section.}

\section{Simulation Results}
\label{section.simulation}

{In this section,
Monte Carlo simulations of data transmission
are used to analyze the average BER performance of the proposed
generalized prefix for two different scenarios.
For each Monte Carlo run,
a channel is generated according to the simulation environment
described in its respective scenario.}
The first scenario is a simple example to illustrate the effectiveness of the
proposed technique.
The proposed prefix in this case transforms the
linear convolution into a skew-circular convolution.
In a second scenario, we present results obtained from a practical
example, where random multipath channels were optimized. In our
simulations, we selected the tolerance factor for the golden section
search as $\epsilon=10^{-3}$. In this case, the golden search
algorithm required less than 10 iterations to find an optimal value
$\psi = e^{j \alpha}$ in the search space $\alpha \in [0, 2\pi/N]$.

\subsection{Example 1}

To demonstrate the advantage of the proposed prefix, first we use a
simple wireless channel with only two paths $H(z) = \left( 1 +
z^{-1} \right)/\sqrt{2}$. Then  $H[k] = \left( 1 + {e^{-j2\pi k/N}}
\right)/\sqrt{2}$. Considering an OFDM system with $N=64$, we notice
that the 32th subcarrier has a spectral null, since $H[32]=0$. This
leads to an error floor of $1/(2\times 64)=0.0078125$.

The use of a cyclic prefix scheme means that $\alpha$ is
$2\pi/N\approx 0.09817477$. But, this frequency domain shift only
moves the existing spectral null from one subcarrier to another and
the system error floor remains unchanged. This simple yet useful
example illustrates {some difficulties} of the cyclic prefix.

To develop some intuition about the significance of the optimum
value of $\psi=e^{j\alpha}$, the bit error rate (BER) performance of
the proposed system is shown in Fig.~\ref{fig:rec3} for different
values of $\alpha$ at selected values of $E_b/N_0 \in \{ 20, 30, 35
\}$~dB. It is shown that the chosen $\psi$ plays an important role
in the BER performance of the proposed system. Moreover, it is
observed that the objective function is unimodal over the range
$\alpha \in [0, 2\pi/N] \simeq [0, 0.1]$, for $N=64$. Also, its
minimum values are numerically found at points of the form $\psi =
e^{j(2m+1)\pi/N}$, where $m$ is an integer. This latter fact is
analytically discussed below. The optimum value numerically obtained
is $\alpha^\ast \approx 0.05$. This value is in agreement with the
overall behavior of the channel frequency response. In fact, given
that there is a single spectral null, which is located at subcarrier
$k=N/2$,
it is expected that the best way to move the null away from the
affected subcarrier is to place it between two adjacent subcarriers.
This can be accomplished by a minimum frequency shift of $|\omega_k
- \omega_{k\pm1}|/2 = \pi/N \approx 0.049087385$. In this case, the
proposed prefix transforms the linear convolution into a
skew-circular convolution. Considering the discussed class of
equivalence, any frequency shift of the form $\alpha^\ast = \pi/N +
2\pi m/N = (2m+1)\pi/N$, for any integer $m$, is {adequate
for our purposes}.

\begin{figure}%
\centering
\epsfig{figure=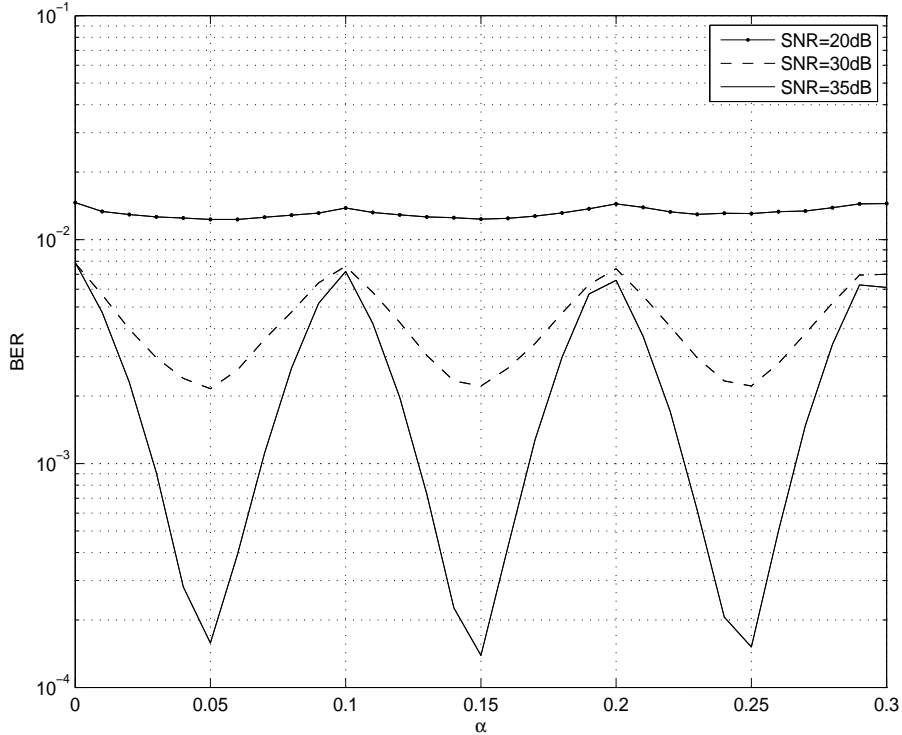,height=10cm,width=12cm}
\caption{BER performance for different $\alpha$ values.}
\label{fig:rec3}
\end{figure}

Finally, the proposed system is compared with CP-OFDM and ZP-OFDM in
Fig.~\ref{fig:rec6} for the following parameters: 4-QAM modulation,
$N=64$, and $K=16$. It is seen that CP-OFDM has an error floor at
about $10^{-2}$. It is also seen that using the proposed prefix
leads to a substantial performance advantage. In particular, it
outperforms ZP-OFDM by about 5~dB at $\text{BER} = 10^{-4}$ (in
addition to having a much lower computational complexity).

\begin{figure}%
\centering
\epsfig{figure=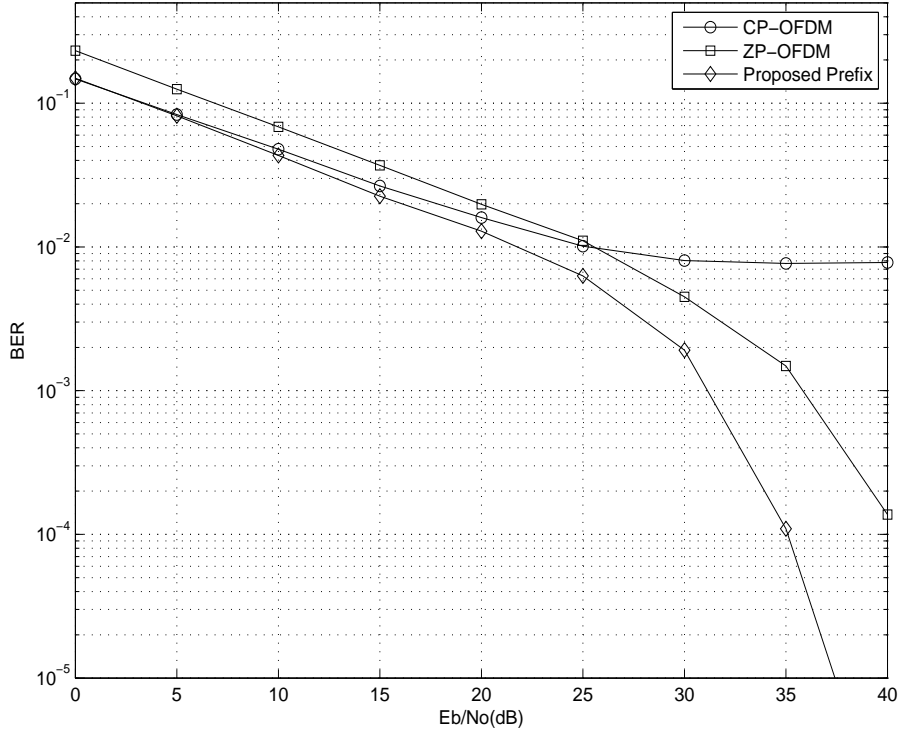,height=10cm,width=12cm}
\caption{Comparison of bit error rates for OFDM with cyclic prefix,
OFDM with zero-padding, and the proposed prefix.}
\label{fig:rec6}
\end{figure}

\subsection{Example 2}

Now we consider a system appropriate for outdoor wireless
communications with a 5~MHz bandwidth, divided into $N=512$ tones.
If we let the prefix length to be $K=64$, the total OFDM symbol
would consist of $N+K = 576$~samples. This implies a total symbol
period
of 115{.}2 $\mu$s, of which 12{.}8 $\mu$s is allocated to the
prefix. {The OFDM frame has a time duration 16.12~ms,
consisting of 20~slots, each 0.806~ms. Each slot consists of seven
OFDM symbols. In the case of block-type channel estimation, the
first OFDM symbol of each slot is allocated as pilot tones.}
The QPSK modulation format is employed. Additionally, the wireless
channel between the transmitting antenna and the receiver antenna is
modeled according to realistic channel specifications determined by
the COST-207 project~\cite[p.~82]{du2010wireless}. {In the
following, quasi-static channel is assumed unless otherwise
specified.} We considered the typical urban (TU) and the bad urban
(BU) channel models with a 12-tap channel model. {For the TU
channel model the delay spread is approximately 1~$\mu$s (-10~dB on
5~$\mu$s). For the BU channel model the delay spread is
approximately 2.5~$\mu$s (-15~dB on 10~$\mu$).} {The power
delay profile (PDP) of the channels that have been used is given in
Fig~\ref{fig:PDP}.}

\begin{figure}%
\centering \epsfig{figure=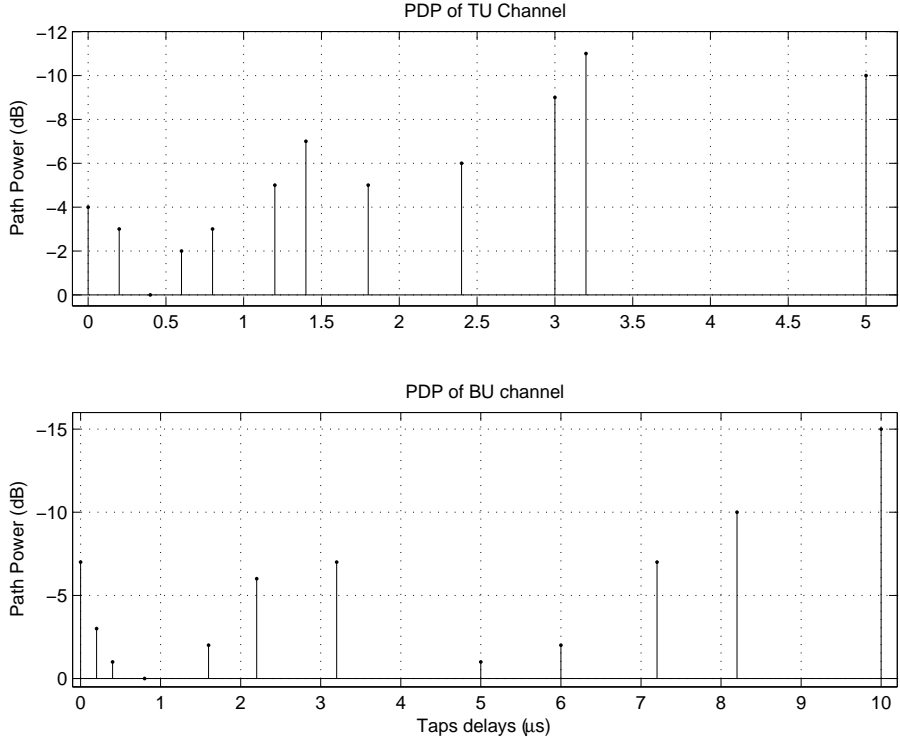,height=10cm,width=12cm}
\caption{Power delay profiles of the COST 207 channels.}
\label{fig:PDP}
\end{figure}

\subsubsection{Perfect CSI}

Assuming that perfect CSI is available at the receiver, the BER of
CP-OFDM and the proposed method is shown in Fig.~\ref{fig:recE1}.

\begin{figure}%
\centering
\epsfig{figure=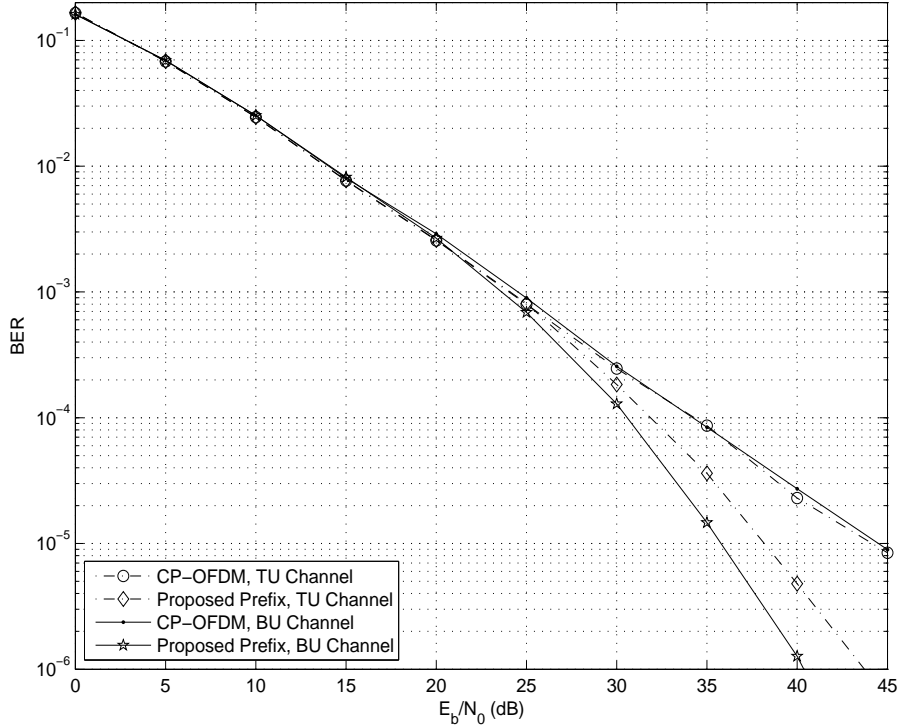,height=10cm,width=12cm}
\caption{BER comparison of CP-OFDM and OP-OFDM over COST 207 channels (12 taps).}
\label{fig:recE1}
\end{figure}

Considering TU channels, the proposed prefix gains about 7~dB over
CP-OFDM at $\text{BER} = 10^{-5}$. For BU channels, the gain at
$\text{BER}=10^{-5}$ is about 9~dB.
These significant performance advantages intuitively can be
explained as follows. Notice that longer delay spreads generally
lead to more deep fades and spectral nulls. Therefore, BU channels
have more deep fades than TU channels. For this reason, the the
performance advantage of the proposed system is grater over BU
channels. Therefore, the proposed prefix construction effectively
mitigates the effects of multipath propagation.

Additionally, the performance difference between the proposed prefix
and CP-OFDM increases as the SNR increases. This is especially
noticeable for TU channels.

\subsubsection{Impact of Channel Estimation}

In the previous subsection perfect CSI was assumed.

In this subsection, we assume that the channel is not known exactly
and channel estimation is necessary.
Here, channel estimation is performed according to one of two
approaches: (i) block-type estimation, which inserts pilot tones
into all of subcarriers; or (ii) comb-type estimation, which
uniformly inserts pilot tones into a OFDM
symbol~\cite{coleri2002estimation}.

In our simulations we consider a quasi-static fading model, which
allows the channel state to be constant during the transmission of a
block consisting of several OFDM symbols~\cite{stamoulis2002fading}.
In the case of block-type estimation,
only the first received OFDM symbol
was employed for the channel estimation and optimization procedures.
Then the optimal $\psi$ was returned to the transmitter.
Such optimized $\psi$ and the estimated channel response were
reused at the receiver for the remaining OFDM symbols in the slot.
The pilot-assisted channel estimation scheme that we use is based on
least square (LS) techniques~\cite{tong2004pilot}. In the case of
comb-type estimation, the pilot insertion rate (PIR) was chosen to
be 1:4. Channel estimation and optimization were performed for each
OFDM symbol. The optimized values of $\psi$ were returned to the
transmitter continuously.
Equalization was done at the receiver by using the current $\psi$ values
and estimated channel response.
Additionally,
a low-pass interpolation was also employed after LS estimation.

Imperfect CSI leads to some performance degradation for all systems
compared to perfect CSI.
 Fig.~\ref{fig:recE2} shows that for CP-OFDM there is a
3~dB loss at $\text{BER}=10^{-4}$ in the case of block-type channel
estimation. In contrast, the loss for the proposed system is only
0.8~dB.

\begin{figure}%
\centering
\epsfig{figure=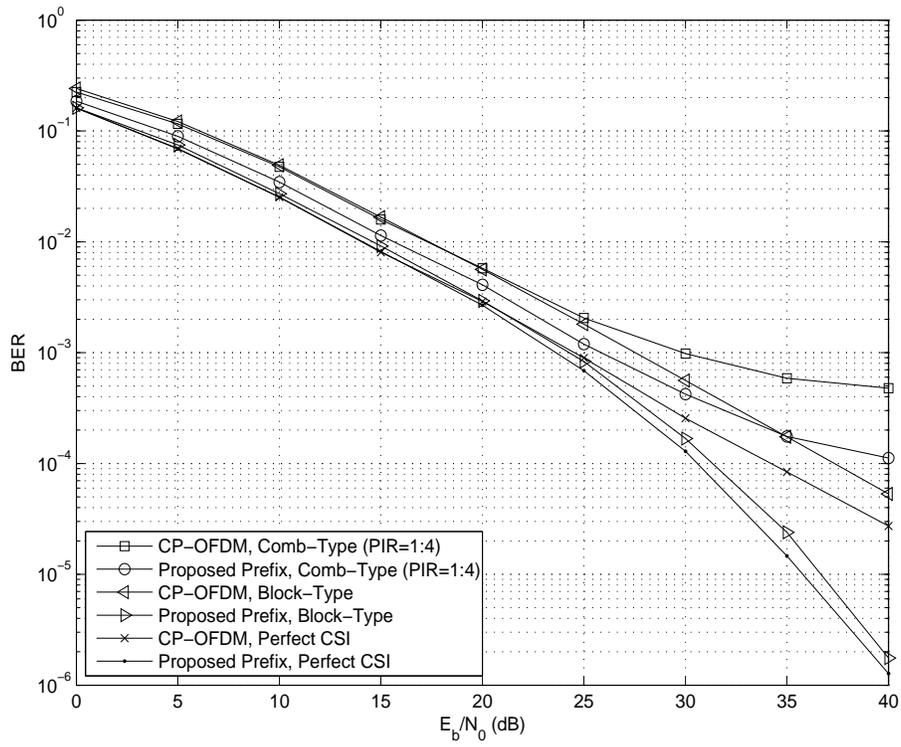,height=10cm,width=12cm}
\caption{Impact of estimated channel impulse response on BER performance of the compared systems for BU channel (12 taps).}
\label{fig:recE2}
\end{figure}

Therefore,
the proposed prefix is
more robust to channel estimation errors
than the CP-OFDM.
Similarly to the results in the previous subsection, the performance
difference between the two considered systems increases as the SNR
is increased, regardless of the type of channel estimation.

It is also of interest to investigate the performance of the
proposed system in the presence of mobility.
{In Fig.~\ref{fig:recE3}, the BER of the CP-OFDM and the
proposed system are shown for mobility of 20~km/h where the Doppler
frequency is equal to 44.44~Hz.}
As expected, both systems have a significant performance loss for
the block-type channel estimation.
However, as expected, for both systems, comb-type estimation schemes
are more appropriate when there is mobility.

\begin{figure}%
\centering
\epsfig{figure=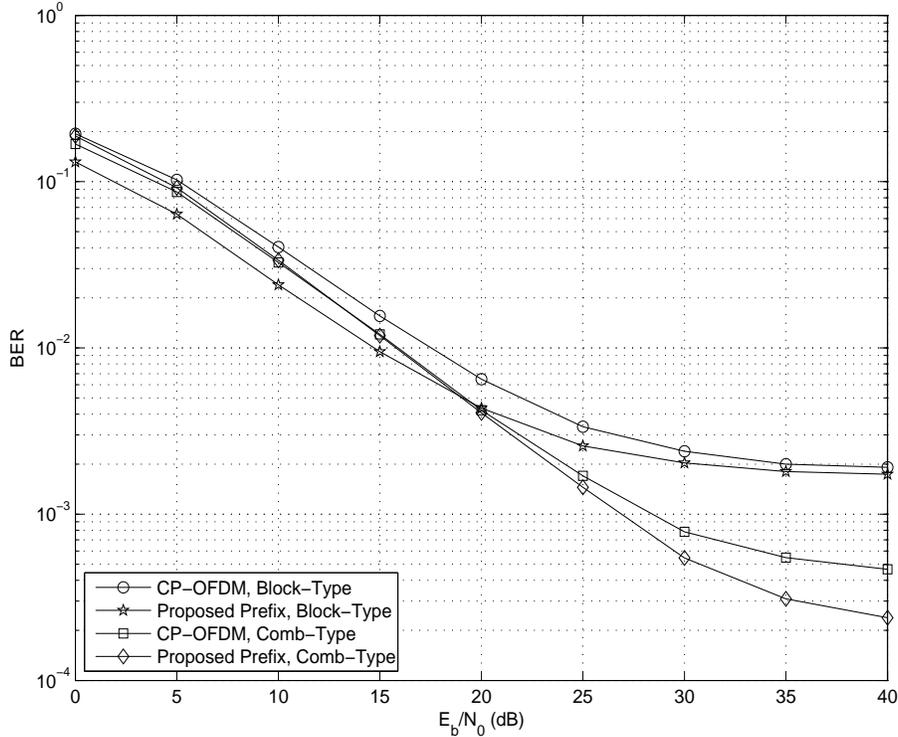,height=10cm,width=12cm}
\caption{Impact of mobility on BER performance of the compared systems for BU channel (12 taps) at 20~km/h.}
\label{fig:recE3}
\end{figure}

\section{Conclusions}
\label{section.conclusions}

In this paper a new prefix for OFDM systems is proposed. The
proposed system has the lowest bit error rate among the other prefix
construction techniques, considering all other parameters identical.
The prefix that is used is not random, it is determined using a
computationally simple optimization routines.
This prefix shifts the phases of the multipath components and
effectively changes the wireless channel experienced by the OFDM
system into a different channel. In particular, a wireless channel
that has deep fades or spectral nulls is transformed into a channel
with fades that are less deep or no spectral nulls. This property of
the proposed system is equivalent to claiming that it ensures
maximum diversity gain~\cite{wang2001linearly}.

The reason for the performance advantages of the proposed technique
can explained intuitively. For OFDM the probability of error is the
average probability of error on all subcarriers, as in (\ref{11}).
In the presence of deep fades at some subcarriers, the probability
of error for these subcarriers is large and will dominate. However,
it is best if the probabilities of error for each subcarrier are
identical or as much as possible similar. Then the probability of
error of the system is minimized. This is what the new technique
achieves. The simulation results indicate that the performance
advantages are significant in all cases of interest. The performance
improvement approaches 10 dB in the vicinity of $\text{BER} =
10^{-5}$ for BU channels. The technique is also superior in the
presence of channel estimation errors and mobility. Also it is
overall computationally simple and allows the use of
cross-correlation or autocorrelation-based synchronization methods.
The new method transforms the linear convolution into a convolution
operation called generalized skew-circular convolution.

This research has further significance for orthogonal
frequency-division multiple access (OFDMA) systems, in which each
user experiences a different channel. In OFDMA systems with dynamic
allocation, the base station informs the users about the assigned
subcarriers~\cite{SunFanglei2009}. In this case, due the dynamic
variation of the number of active users, the optimization is
computationally complex and the bandwidth consumed by the feedback
channel is significant~\cite{song2005utility}. Static allocation is
simpler, but some of the assigned subcarriers may be experiencing
deep fades. According to the proposed technique each user can apply
an individual generalized prefix~\cite{khattab2006opportunistic}.
Therefore, the system performance can be increased for both static
and dynamic allocation systems. This is left as a topic for future
research.

\section*{Acknowledgments}

This work was supported in part
by Raytheon Corp., Fort Wayne, IN, through SERC grant 204135,
by the Research Fund of Istanbul University under project YADOP-6265,
by \emph{Conselho Nacional de Desenvolvimento Cient\'ifico e Tecnol\'ogico} (CNPq)
and
by FACEPE, Brazil.

{\small
\bibliographystyle{IEEEtran}
\bibliography{ref}
}

\end{document}